\documentclass[letterpaper, 10pt, conference]{ieeeconf}      
\IEEEoverridecommandlockouts
\def\BibTeX{{\rm B\kern-.05em{\sc i\kern-.025em b}\kern-.08em
    T\kern-.1667em\lower.7ex\hbox{E}\kern-.125emX}}

\usepackage{color,url}

\usepackage{amsmath,amssymb,amsfonts}
\usepackage{graphicx}
\usepackage{textcomp}

\usepackage{amsthm}

\let\labelindent\relax
\usepackage{enumitem}

\usepackage{bm}
\usepackage{soul}
\usepackage{comment}

\usepackage[linesnumbered,ruled,vlined]{algorithm2e}
\SetKwInput{KwInput}{Input}
\SetKwInput{KwOutput}{Output}

\newtheorem{theorem}{Theorem}

\newtheorem{corollary}{Corollary}
\newtheorem{lemma}{Lemma}

\newtheorem{definition}{Definition}

\allowdisplaybreaks

\usepackage{cite}
\usepackage{accents}

\usepackage[colorlinks=true,linkcolor=black]{hyperref}
\hypersetup{
  colorlinks,
  citecolor=black,
  linkcolor=black,
  urlcolor=black}

\begin{document}
\title{
{\fontsize{21pt}{1pt}\selectfont
Matrix Concentration Inequalities for Sensor Selection}
}

\author{Christopher I. Calle and  Shaunak D. Bopardikar
\thanks{The authors are with the Department of Electrical and Computer Engineering, Michigan State University (MSU), East Lansing, MI 48824, USA. Emails: \texttt{callechr@msu.edu, shaunak@egr.msu.edu} }
\thanks{This work was supported in part by NSF Grant \# ECCS-2030556, GAANN Grant \# P200A180025, the National GEM Fellowship, and MSU's University Enrichment Fellowship (UEF).}
}

\maketitle

\begin{abstract}
In this work, we address the problem of sensor selection for state estimation via Kalman filtering. We consider a linear time-invariant (LTI) dynamical system subject to process and measurement noise, where the sensors we use to perform state estimation are randomly drawn according to a sampling with replacement policy. Since our selection of sensors is randomly chosen, the estimation error covariance of the Kalman filter is also a stochastic quantity. Fortunately, concentration inequalities (CIs) for the estimation error covariance allow us to gauge the estimation performance we intend to achieve when our sensors are randomly drawn with replacement. To obtain a non-trivial improvement on existing CIs for the estimation error covariance, we first present novel matrix CIs for a sum of independent and identically-distributed (i.i.d.) and positive semi-definite (p.s.d.) random matrices, whose support is finite. Next, we show that our guarantees generalize an existing matrix CI. Also, we show that our generalized guarantees require significantly fewer number of sampled sensors to be applicable. Lastly, we show through a numerical study that our guarantees significantly outperform existing ones in terms of their ability to bound (in the semi-definite sense) the steady-state estimation error covariance of the Kalman filter.
\end{abstract}

\begin{keywords}
Concentration inequalities, Kalman filtering, Random matrix theory, Sensor selection
\end{keywords}

\section{Introduction}
\label{section:introduction}

The selection of sensors for state estimation is a mature subject in the literature. Typically, the sensors under consideration, referred to as candidate sensors in this work, are distinct and chosen without replacement. However, it turns out that choosing sensors with replacement is also useful for approximating several scenarios of interest. In the area of state estimation for dynamical systems, sampling a sensor with replacement can be interpreted in at least two different ways. Before we describe the two types of interpretations and their practical relevance, we first clarify that sampling with replacement is a type of sampling policy that draws each sensor at random and \textit{with replacement} from a pool of candidate sensors via some specified probability distribution. We formally introduce the sampling policy in Section~\ref{subsection:sampling_policy}.

In the first interpretation, sampling a candidate sensor with replacement equates to acquiring multiple measurements from that same sensor at a specific time instant $t$. Though acquiring these types of measurements is not physically possible, this interpretation is still useful in approximating a practical scenario of interest. For example, if the rate at which a sensor can perform measurements is significantly greater than the rate at which the dynamical system under consideration changes, then sampling that candidate sensor multiple times approximately equates to the act of acquiring multiple measurements from that one sensor at distinct time instances (that neighbor the nominal time instant $t$).

In the second interpretation, sampling a candidate sensor with replacement equates to sampling without replacement from a redundant sensor network (RSN). We refer to a pool of candidate sensors, where multiple copies of that same sensor exist, as an RSN. However, it may not always be possible for a sampling pool to possess multiple copies of each sensor. If the candidate sensors of the original sampling pool can be allocated into groups of similar sensors, then we refer to the sampling pool as an approximate RSN. As a consequence of the second interpretation, sampling with replacement approximately equates to sampling without replacement from an approximate RSN.

If either interpretation applies to a sensor selection problem for state estimation, then  it was shown in \cite{10193790} that concentration inequalities (CIs) are ideal tools for bounding the state estimation error covariance of the Kalman filter applied to a linear time-invariant (LTI) dynamical model.


Historically, the guarantees of CIs were only applicable to sums of scalar-valued random variables. Fortunately, the influential work by Ahlswede and Winter (AW) in \cite{ahlswede2002strong} allowed for the natural extension of CIs to sums of matrix-valued random variables. Comprehensive works, such as \cite{tropp2015introduction}, typically express CIs for the spectral norm of a sum of random matrices. However, as shown in \cite{qiu2014cognitive} and \cite{notes13}, CIs that bound (in the semi-definite sense) a sum of random matrices also exist. In this work, we refer to the latter type of CIs as matrix CIs since they bound the spectrum of a sum instead of just the spectral norm of a sum. The ability to provide probabilistic guarantees for the entire spectrum of a sum of random matrices is one appealing feature of matrix CIs. It is these types of CIs that \cite{10193790} used to capture the spectrum of the estimation error covariance of the Kalman filter.


The use of matrix CIs in the field of control theory is a relatively unexplored topic. Only a handful of works, such as \cite{bopardikar2021randomized,10193790,siami2020deterministic,amini2022space}, have recognized the utility of sums of random variables to address initial state estimation, actuator scheduling, and state estimation. The reason why works, such as \cite{bopardikar2021randomized} and \cite{10193790}, could apply matrix CIs to the problem of state estimation was due to their focus on sampling \textit{with replacement} policies. 


As mentioned earlier, the type of sampling we employ in this work is in contrast to mature and contemporary works, such as \cite{clark2016submodularity,zhang2017sensor,hashemi2021randomized}, which do not assume a sensor can be drawn more than once and rely on greedy sampling policies to choose an appropriate selection of sensors for state estimation. Though a greedy algorithm can provide assurances on the quality of the selection relative to the optimal selection, the utility of such assurances ultimately depends on the submodularity of the problem. By focusing on sampling \textit{with replacement} policies, the work \cite{10193790} was able to provide useful guarantees without regard to the submodularity of the problem. This is yet another reason why matrix CIs for the state estimation error covariance (of the Kalman filter) are of great interest.

In this work, we focus on generalizing an existing AW-based matrix CI, referred to as Lemma~\ref{lemma:AW_inequality} in this work, to improve upon existing bounds for the estimation error covariance of the Kalman filter. This work specifically explores the extent to which we can outperform the existing concentration-based guarantees of \cite{10193790}. Though the work \cite{10193790} also explores how to choose an optimal sampling policy, we omit this discussion in this work. However, we do leverage the optimal sampling policy to help compare our guarantees with those of \cite{10193790}. The goal of this work is to show how our generalization of an existing CI can lead to tighter bounds for the estimation error covariance of the Kalman filter.

{\textit{Contributions:}}
Our main contributions are three-fold. First, we derive novel matrix CIs for a sum of independent and identically-distributed (i.i.d.) and positive semi-definite (p.s.d.) random matrices. We show that our CIs are a generalization of an existing AW-based matrix CI. Second, we show that our generalized guarantees require significantly fewer number of sampled sensors. Third, we apply our guarantees to the problem of sensor selection for state estimation. We compare our CIs for the steady-state error covariance of the Kalman filter against the existing guarantees of \cite{10193790}, and we empirically show that our CIs outperform them.

\section{Preliminaries}
\label{section:problem_formulation}

\subsection{Notation}
\label{subsection:notation}

We denote $\mathbb{R}$ and $\mathbb{N}$ as the set of real and natural numbers, respectively, and $I_n$ and $0_n$ as the identity and zero matrix of order $n \times n$, respectively. Also, we denote $\bar{0}_n$ as a column vector with $n$ zero entries. Furthermore, we denote $\overline{\lambda} ( \cdot )$ as the maximum eigenvalue of a square matrix argument, $\Delta^{n}$ as the probability simplex in $\mathbb{R}^n$, and $\| \cdot \|$ as the spectral norm of a matrix argument. We refer to $\mathbb{S}^{n}$, $\mathbb{S}_{+}^{n}$, and $\mathbb{S}_{++}^{n}$ as the set of symmetric, positive semi-definite (p.s.d.), and positive definite (p.d.) matrices of order $n \times n$, respectively. The operators $\succeq$ and $\succ$ apply to any matrices $A, B \in \mathbb{S}^{n}$, where the following inequalities, $A \succeq B$ and $A \succ B$, hold if $A-B \in \mathbb{S}_{+}^{n}$ and $A-B \in \mathbb{S}_{++}^{n}$, respectively. The former and latter inequalities are said to hold in the semi-definite and definite sense, respectively. Also, we denote $A^{+}$ as the pseudo-inverse of matrix $A \in \mathbb{R}^{n \times m}$. The notation $[n]$ is shorthand for the set $\{ 1, 2, \ldots , n \}$.

\subsection{Statistical Properties}
\label{subsection:properties}

In this work, we focus on CIs that probabilistically bound (in the semi-definite sense) a sum of p.s.d. random matrices. We denote $\boldsymbol{Z}_i$ as the $i$-th copy of the matrix-valued random variable $\boldsymbol{Z} \in \mathbb{S}_{+}^{d}$ and $( \boldsymbol{Z}_i )_{i \in [n]}$ as a sequence of $n$ i.i.d. and p.s.d. random matrices. Throughout this paper, we designate the variable $d \in \mathbb{N}$ to refer to the dimension of the p.s.d. random matrix $\boldsymbol{Z}$. Also, we denote $\eta \in \mathbb{N}$ as the number of unique realizations of $\boldsymbol{Z}$, $\mathcal{Z}_i \in \mathbb{S}_{+}^{d}$ as the $i$-th realization of $\boldsymbol{Z}$, and $\mathsf{S}_z \hspace{-0.25mm} := \hspace{-0.25mm} \{ \mathcal{Z}_i \}_{\, i \in [ \eta ]}$ as the support of $\boldsymbol{Z}$. Lastly, we denote $p \in \Delta^{\eta}$ as the probability distribution of $\boldsymbol{Z}$. Observe that the expectation $\mathbb{E}[ \boldsymbol{Z} ]$ is equivalent to ${\textstyle\sum\nolimits_{\, i \in [ \eta ]}} \, p_i \hspace{0.25mm} \mathcal{Z}_i$.

\subsection{Sampling Policy}
\label{subsection:sampling_policy}

We consider a sampling pool of $\eta$ candidate sensors. Each candidate sensor is assigned an integer that ranges from 1 to $\eta$. Without loss of generality, the $i$-th candidate sensor is assigned the integer $i$ and identified by it. Our sampling policy consists of $\gamma$ independent draws from a categorical random variable, whose probability distribution is $p$ and support is [$\eta$]. Each sample is stored as an entry in the sequence $\mathcal{S} \in [\eta]^{\gamma}$, i.e., the entry $\mathcal{S}_i$ is the realization of the $i$-th independent draw from the categorical random variable. As a consequence, the selection $\mathcal{S}$ of sensors is also a stochastic quantity. 

We observe that the sampling distribution $p$ is sufficient for specifying our sampling policy. As a consequence, we use the following terms, sampling distribution and sampling policy, interchangeably. Also, since each candidate sensor can to be drawn multiple times, we say that our policy is a sampling \textit{with replacement} policy. Also, we use a selection $\mathcal{S}$ of sensors to sample a sequence $( \boldsymbol{Z}_i )_{i \in [\gamma]}$ of random variables, where $\boldsymbol{Z}_i := \{ \, \mathcal{Z}_j : j = \mathcal{S}_i \, \}$ for all $i \in [\gamma]$.

\subsection{Dynamical System Model}

We consider a linear time-invariant (LTI) model corrupted by process and measurement noise, i.e., 
\begin{align}
\label{eqn:system_model_lti}
\begin{split}
x_{(t+1)} &= A \, x_{(t)} + w_{(t)},   \\
y_{(t)} &= C_{\mathcal{S}} \, x_{(t)} + v_{\mathcal{S},(t)},
\end{split}
\end{align}
where $x_{(t)} \in \mathbb{R}^{d}$ and $y_{(t)} \in \mathbb{R}^{\gamma}$ denote the state and output vector at time instant $t$, respectively. In this work, we also designate $d$ as the state dimension of \eqref{eqn:system_model_lti}. We denote $A \in \mathbb{R}^{d \times d}$ as the state matrix and $C_{\mathcal{S}} \in \mathbb{R}^{\gamma \times d}$ as the output matrix. We also denote $w_{(t)} \in \mathbb{R}^{d}$ and $v_{\mathcal{S},(t)} \in \mathbb{R}^{\gamma}$ as the process and measurement noise, respectively. We specify the quantities $\hat{x}_{f,(-1)} \in \mathbb{R}^{d}$ and $P_{\mathcal{S},(-1)} \in \mathbb{S}_{+}^{d}$. Also, we assume $x_{(0)} \sim \mathcal{N}( \hat{x}_{f,(-1)},P_{\mathcal{S},(-1)})$, $w_{(t)} \sim \mathcal{N}(\bar{0}_d,Q)$, and $v_{\mathcal{S},(t)} \sim \mathcal{N}(\bar{0}_d,R_{\mathcal{S}})$ are Gaussian random variables. We refer to $Q$ and $R_{\mathcal{S}}$ as the process and measurement noise covariance matrices, respectively. Lastly, we assume $x_{(0)}$, $\{ w_{(t)} \}_{t=0}^{\infty}$, and $\{ v_{\mathcal{S},(t)} \}_{t=0}^{\infty}$ are mutually uncorrelated.

Next, we characterize a candidate sensor in the context of the LTI system \eqref{eqn:system_model_lti}. Observe that the $i$-th candidate sensor is uniquely identified by the pair $(\bm{c}_{i},\bm{\sigma}_{i}^{2})$, where $\bm{c}_{i} \in \mathbb{R}^{d}$ maps the state to the uncorrupted output and $\bm{\sigma}_{i}^{2} > 0$ is the variance of the Gaussian noise corrupting the output of the $i$-th candidate sensor.

Next, we show how the output matrix $C_{\mathcal{S}} \in \mathbb{R}^{\gamma \times d}$ and measurement noise covariance $R_{\mathcal{S}} \in \mathbb{S}_{++}^{\gamma}$ are constructed using the selection $\mathcal{S}$ of sensors. We denote $c_{i}^{T}$ as the $i$-th row of $C_{\mathcal{S}}$ and $\sigma_{i}^{2}$ as the $i$-th diagonal element of $R_{\mathcal{S}}$. Since each sensor of LTI system~\eqref{eqn:system_model_lti} is independent of each other, the matrix $R_{\mathcal{S}}$ is diagonal. We assign the $i$-th sensor of \eqref{eqn:system_model_lti} using to the selection~$\mathcal{S}$ of sensors, i.e., 
\begin{align*}
( c_{i} , \sigma_{i}^{2} ) := \{ \, (\bm{c}_{j},\bm{\sigma}_{j}^2) : j = \mathcal{S}_{i} \, \}, \ \forall i \in [\gamma].
\end{align*}

\subsection{Kalman filtering}
\label{subsection:kalman_filtering}

We use the Kalman filter to compute an estimate $\hat{x}_{(t)}$ of the actual state $x_{(t)}$ of \eqref{eqn:system_model_lti}. When no measurement $y_{(t)}$ is available at time instant $t$, we designate $\hat{x}_{p,(t)}$ as the state estimate of $x_{(t)}$ and $\Sigma_{\mathcal{S},(t)}$ as the predicted covariance matrix of the state estimation error $( \hat{x}_{p,(t)} - x_{(t)} )$. When a measurement $y_{(t)}$ is available  at time instant $t$, we designate $\hat{x}_{f,(t)}$ as the state estimate of $x_{(t)}$ and $P_{\mathcal{S},(t)}$ as the filtered covariance matrix of the state estimation error $( \hat{x}_{f,(t)} - x_{(t)} )$. Throughout this work, we assume a measurement $y_{(t)}$ is available at each time instant. We use the following equations to update the predicted and filtered covariance matrices,
\begin{gather*}
\Sigma_{\mathcal{S},(t)} = A P_{\mathcal{S},(t-1)} A^{T} + Q, \\
P_{\mathcal{S},(t)} = \Sigma_{\mathcal{S},(t)} \hspace{-0.5mm} - \hspace{-0.5mm} \Sigma_{\mathcal{S},(t)} C_{\mathcal{S}}^{T} ( R_{\mathcal{S}} \hspace{-0.25mm} + \hspace{-0.25mm} C_{\mathcal{S}} \Sigma_{\mathcal{S},(t)} C_{\mathcal{S}}^{T} )^{-1}  C_{\mathcal{S}} \Sigma_{\mathcal{S},(t)},
\end{gather*}
for all $t \geq 0$. If $Q$ and $R_{\mathcal{S}}$ are p.d. matrices, then the above equations can be compactly stated in terms of the filtered error covariance at two neighboring time instances, i.e.,
\begin{align}
\label{eqn:P_filtered}
P_{\mathcal{S},(t+1)} = ( \, ( \, A P_{\mathcal{S},(t)} A^T + Q \, )^{-1} + C_{\mathcal{S}}^{T} R_{\mathcal{S}}^{-1} C_{\mathcal{S}} \, )^{-1}
\end{align}
for all $t \geq -1$. Also, if the pairs $(A,Q^{1/2})$ and $(A,C_{\mathcal{S}})$ are stabilizable and detectable, respectively, then we define the unique p.d. steady-state solution of the recursive covariance equation \eqref{eqn:P_filtered} as $P_{\mathcal{S}} := \lim_{t \rightarrow \infty} P_{\mathcal{S},(t)}$
given that the sequence starts at $P_{\mathcal{S},(-1)} \in \mathbb{S}_{+}^{d}$. Throughout this work, we assume $Q$ is a p.d. matrix, i.e., $Q \in \mathbb{S}_{++}^{d}$. Also, observe that $R_{\mathcal{S}}$ is a diagonal and p.d. matrix, i.e., $R_{\mathcal{S}} \in \mathbb{S}_{++}^{\gamma}$, since each sensor in \eqref{eqn:system_model_lti} is independent of each other and has non-zero measurement noise variance.

Next, we show that the product $C_{\mathcal{S}}^{T} R_{\mathcal{S}}^{-1} C_{\mathcal{S}}$ in the recursive equation \eqref{eqn:P_filtered} can be expressed as a finite sum of p.s.d. random matrices. First, we construct $\boldsymbol{Z}_i := \sigma_{i}^{-2} c_{i} c_{i}^{T}$ for all $i \in [\gamma]$ and $\mathcal{Z}_i := \bm{\sigma}_{i}^{-2} \bm{c}_{i} \bm{c}_{i}^{T}$ for all $i \in [\eta]$. Second, observe that $\boldsymbol{Z}_i$ is a p.s.d. random matrix since the sequence $\mathcal{S}$ is randomly drawn via the sampling policy in Section~\ref{subsection:sampling_policy}. Third, observe that the definitions of $C_{\mathcal{S}}$ and $R_{\mathcal{S}}$ imply the following,
\begin{align*}
C_{\mathcal{S}}^{T} R_{\mathcal{S}}^{-1} C_{\mathcal{S}} = \textstyle\sum\nolimits_{\hspace{0.25mm} i \in [\gamma]} \boldsymbol{Z}_{i} = \textstyle\sum\nolimits_{\hspace{0.25mm} i \in \mathcal{S}} \mathcal{Z}_{i}.
\end{align*}

Lastly, we introduce the matrix function $f$ in Definition~\ref{def:f} to compactly state equations throughout this work.
\begin{definition}
\label{def:f}
Given $\Lambda \in \mathbb{S}_{++}^{d}$ and $\Xi \in \mathbb{S}_{+}^{d}$, we define~the~matrix function $f( \Lambda , \Xi ) := ( \, ( \, A \Lambda A^T + Q \, )^{-1} + \Xi \, )^{-1}$.
\end{definition}

\subsection{Problem Statement}

In this work, we derive matrix CIs to bound (in the semi-definite sense) the estimation error covariance of the Kalman filter when applied to the LTI system~\eqref{eqn:system_model_lti}. We focus on the steady-state solution of the estimation error covariance. In \cite{10193790}, an existing CI (referred to as Lemma~\ref{lemma:AW_inequality} in this work) is used to derive a CI for the steady-state error covariance (referred to as Lemma~\ref{lemma:SS_bounds} in this work). The problem we seek to answer is whether the matrix CI for the steady-state error covariance of the Kalman filter can be substantially improved upon by generalizing the guarantees of Lemma~\ref{lemma:AW_inequality}.

\section{Main Results}
\label{section:main_results}

\subsection{Matrix Concentration Inequalities}
\label{subsection:matrix_CIs}

In this section, we present matrix CIs that probabilistically bound (in the semi-definite sense) a sum of i.i.d. and p.s.d. random matrices. Observe that Lemma~\ref{lemma:AW_inequality} and Theorem~\ref{theorem:inequality_Z} have a few quantities in common, such as $d$, $\gamma$, and $p$. This allows us to properly compare them in our analysis. We refer the reader to Section~\ref{subsection:sampling_policy} on how to obtain a realization for the sequence $( \boldsymbol{Z}_i )_{i \in [\gamma]}$ of random variables. Also, we refer to the quantities that upper or lower bound the sum ${\textstyle\sum\nolimits_{\, i \in [ \gamma ]}} \boldsymbol{Z}_i$ as semi-definite bounds.

First, we introduce the matrix CI in Lemma~\ref{lemma:AW_inequality}, the CI that our guarantees in \cite{10193790} are based off of. Note that Lemma~\ref{lemma:AW_inequality} is a reformulation of Corollary~2 in \cite{notes13}.

\begin{lemma}
\label{lemma:AW_inequality}
Let $( \boldsymbol{Z}_i )_{i \in [\gamma]}$ denote a sequence of $\gamma$ i.i.d. and p.s.d. random matrices. Let the tuple $\bar{\mathcal{T}} \hspace{-0.25mm} := \hspace{-0.25mm} ( d, \bar{\delta}, \gamma, \bar{\rho}, \bar{\epsilon}, p )$ of~parameters satisfy the inequality 
\begin{align}
\label{eqn:inequality_R_bar}
\boldsymbol{Z} \preceq \bar{\rho} \, \mathbb{E}[ \boldsymbol{Z} ]
\end{align}
almost surely for the scalar $\bar{\rho} \geq 1$ and the equality
\begin{align}
\label{eqn:equality}
\gamma \, \bar{\epsilon}^2 / \bar{\rho} = 4 \hspace{0.25mm} \log_{e}{( 2d / \bar{\delta} )}
\end{align}
for the scalars $\bar{\delta}, \bar{\epsilon} \in (0,1)$. Then, the event
\begin{align}
\label{eqn:event_00}
\left\{ \, ( \hspace{0.25mm} 1 - \bar{\epsilon} \hspace{0.25mm} ) \, \gamma \, \mathbb{E}[ \boldsymbol{Z} ] \preceq \textstyle\sum\nolimits_{\, i \in [ \gamma ]} \boldsymbol{Z}_i \preceq ( \hspace{0.25mm} 1 + \bar{\epsilon} \hspace{0.25mm} ) \, \gamma \, \mathbb{E}[ \boldsymbol{Z} ] \, \right\}
\end{align}
occurs with probability at least $(\hspace{0.25mm} 1 - \bar{\delta} \hspace{0.25mm})$.
\end{lemma}

We introduce the function $r( \rho , \zeta ) := ( \hspace{0.25mm} 1 - \zeta^2 / \rho \hspace{0.25mm} )$ to compactly state our matrix CIs in the subsequent theorem, where $\zeta$ and $\rho$ are scalars that are we introduce in Theorem~\ref{theorem:inequality_Z}. Observe that $r( \rho , \zeta )$ is defined if $\rho \neq \zeta^2$. Also, if $\rho > \zeta^2$, then $r(\rho,\zeta) > 0$. Throughout this work, we also refer to $r(\rho,\zeta)$ as the quantity $r$ for brevity in notation. Next, we introduce the matrix CIs in Theorem~\ref{theorem:inequality_Z}, a generalization of the guarantees of Lemma~\ref{lemma:AW_inequality}.

\begin{theorem}
\label{theorem:inequality_Z}
Let $( \boldsymbol{Z}_i )_{i \in [ \gamma ]}$ denote a sequence of $\gamma$ i.i.d. and p.s.d. random matrices. Let the tuple $\mathcal{T} \hspace{-0.75mm} := \hspace{-0.75mm} ( d, \delta, \gamma, \rho, \epsilon, p, \zeta )$ of~parameters satisfy the inequality
\begin{align}
\label{eqn:inequality_R}
\boldsymbol{Z} \preceq \rho \, \mathbb{E}[ \boldsymbol{Z} ]
\end{align}
almost surely for the scalar $\rho \geq 1$ and the equality
\begin{align}
\label{eqn:inequality_theorem_Z_c3}
r \, \gamma \, \epsilon^2 / \rho = 4 \hspace{0.25mm} \log_{e}{( d / \delta )}
\end{align}
for the scalars $\delta \in (0,1)$, $\epsilon \in (0,2]$, and $\zeta \in [0,1]$. Assume the scalars $\rho$ and $\zeta$ satisfy the inequality $\rho > \zeta^2$. Then,
\begin{align}
\label{eqn:inequality_theorem_Z_idd_c4}
\mathbb{P} \left[ \, ( \hspace{0.25mm} 1 - r \epsilon \hspace{0.25mm} ) \, \gamma \, \mathbb{E}[ \boldsymbol{Z} ] \preceq {\textstyle\sum\nolimits_{\, i \in [ \gamma ]}} \, \boldsymbol{Z}_i \, \right] &\geq (\, 1-\delta \,), \\
\label{eqn:inequality_theorem_Z_idd_c5}
\mathbb{P} \left[ \, {\textstyle\sum\nolimits_{\, i \in [ \gamma ]}} \, \boldsymbol{Z}_i \preceq ( \hspace{0.25mm} 1 + r \epsilon \hspace{0.25mm} ) \, \gamma \, \mathbb{E}[ \boldsymbol{Z} ] \, \right] &\geq (\, 1-\delta \,).
\end{align}
\end{theorem}

We now describe how the matrix CIs in Theorem~\ref{theorem:inequality_Z} differ from the guarantees in Lemma~\ref{lemma:AW_inequality}. The first major difference is the embedding of a refinement factor $\zeta$ throughout the conditions required to use Theorem~\ref{theorem:inequality_Z}. Observe that the presence of $\zeta$ introduces further complexity to the constraints that need to be satisfied. The second major difference is that the semi-definite bounds in \eqref{eqn:inequality_theorem_Z_idd_c4} and \eqref{eqn:inequality_theorem_Z_idd_c5} are one-sided.

Next, we show that the requirements and guarantees of~Theorem~\ref{theorem:inequality_Z} are a generalization of Lemma~\ref{lemma:AW_inequality} by proving the former implies the latter. Before we provide a sketch of the proof, we define $\Phi$ as the intersection of the events in \eqref{eqn:inequality_theorem_Z_idd_c4} and \eqref{eqn:inequality_theorem_Z_idd_c5}, i.e.,
\begin{align*}
\Phi \hspace{-0.25mm} := \hspace{-0.25mm} \left\{ (\hspace{0.25mm} 1 - r \epsilon \hspace{0.25mm} ) \, \gamma \, \mathbb{E}[ \boldsymbol{Z} ] \preceq {\textstyle\sum\nolimits_{ \, i \in [ \gamma ] }} \boldsymbol{Z}_i \preceq (\hspace{0.25mm} 1 + r \epsilon \hspace{0.25mm}) \, \gamma \, \mathbb{E}[ \boldsymbol{Z} ] \right\}.
\end{align*}
The proof first consists of setting the refinement factor $\zeta$ in Theorem~\ref{theorem:inequality_Z} to zero. Next, we combine the guarantees of \eqref{eqn:inequality_theorem_Z_idd_c4} and \eqref{eqn:inequality_theorem_Z_idd_c5} to derive the following matrix CI, $\mathbb{P}[\Phi] \geq ( 1 - 2 \delta )$, by leveraging the addition rule of probability for two non-mutually exclusive events. Then, we restrict the interval of $\epsilon$ from $(0,2]$ to $(0,1)$. Finally, by setting the following~quantities, $2 \delta$, $\epsilon$, and $\rho$, as $\bar{\delta}$, $\bar{\epsilon}$, and $\bar{\rho}$, respectively, we derive the requirements and guarantees of Lemma~\ref{lemma:AW_inequality}.

Next, we comment on the when the semi-definite bounds of Lemma~\ref{lemma:AW_inequality} and Theorem~\ref{theorem:inequality_Z} are non-trivial. We refer to the bounds of either Lemma~\ref{lemma:AW_inequality} or Theorem~\ref{theorem:inequality_Z} as non-trivial if they are non-zero p.s.d. matrices and trivial if they are zero or negative semi-definite matrices. 
First, observe that the bounds in \eqref{eqn:event_00} are always non-trivial since the following requirements, $\bar{\epsilon} \in (0,1)$ and $\gamma \in \mathbb{N}$, of Lemma~\ref{lemma:AW_inequality} imply $( \hspace{0.25mm} 1 - \bar{\epsilon} \hspace{0.25mm} ) \, \gamma > 0$. The lower bound $( \hspace{0.25mm} 1 - \bar{\epsilon} \hspace{0.25mm} ) \, \gamma \, \mathbb{E}[ \boldsymbol{Z} ]$ is non-trivial since $( \hspace{0.25mm} 1 - \bar{\epsilon} \hspace{0.25mm} ) \, \gamma > 0$. The upper bound $( \hspace{0.25mm} 1 + \bar{\epsilon} \hspace{0.25mm} ) \, \gamma \, \mathbb{E}[ \boldsymbol{Z} ]$ is non-trivial as a consequence. Similarly, the upper bound in \eqref{eqn:inequality_theorem_Z_idd_c5} is always non-trivial since the following requirements, $\rho > \zeta^2$, $\epsilon \in (0,2]$, and $\gamma \in \mathbb{N}$, of Theorem~\ref{theorem:inequality_Z} imply $( \hspace{0.25mm} 1 + r \epsilon \hspace{0.25mm} ) \hspace{0.25mm} \gamma > 0$. In contrast, the lower bound in \eqref{eqn:inequality_theorem_Z_idd_c4} is trivial unless $( \hspace{0.25mm} 1 - r \epsilon \hspace{0.25mm} ) \, \gamma > 0$. If the conditions required to use Theorem~\ref{theorem:inequality_Z} are met, then the inequality $( \hspace{0.25mm} 1 - r \epsilon \hspace{0.25mm} ) \, \gamma > 0$ is equivalent to
\begin{align}
\label{eqn:complexity_00}
\gamma > ( \hspace{0.25mm} \rho - \zeta^2 \hspace{0.25mm} ) \, 4 \hspace{0.25mm} \log_e{ ( d / \delta ) }
\end{align}
as a result of simple algebraic manipulations. Thus, the lower bound in \eqref{eqn:inequality_theorem_Z_idd_c4} is only non-trivial for a range of $\gamma$.

Next, we comment on the sample complexities of the matrix CIs in this section. Note that $\gamma$ denotes the number of random matrices that need to be sampled for Lemma~\ref{lemma:AW_inequality} and Theorem~\ref{theorem:inequality_Z} to be applicable. First, observe that the following requirements, \eqref{eqn:equality}, $\bar{\rho} \geq 1$, and $\bar{\epsilon} \in (0,1)$, of Lemma~\ref{lemma:AW_inequality} imply
\begin{align}
\label{eqn:complexity_01}
\gamma > \bar{\rho} \, 4 \log_e( 2d / \bar{\delta} ) =: \bar{\kappa}.
\end{align}
Similarly, observe that the following requirements, \eqref{eqn:inequality_theorem_Z_c3}, $\rho \geq 1$, $\rho > \zeta^2$, and $\epsilon \in (0,2]$, of Theorem~\ref{theorem:inequality_Z} imply
\begin{align}
\label{eqn:complexity_02}
\gamma \geq \frac{\rho^2}{\, \rho - \zeta^2 \,} \, 2 \log_e( d / \delta ) =: \kappa.
\end{align}
Before we compare the above sampling complexities, we first observe that certain parameters in Lemma~\ref{lemma:AW_inequality} and Theorem~\ref{theorem:inequality_Z} can be equated to each other. For example, if we want to properly compare the matrix CIs in Lemma~\ref{lemma:AW_inequality} and Theorem~\ref{theorem:inequality_Z}, then we must set $\bar{\delta}$ equal to $\delta$. Also, since $\bar{\rho}$ and $\rho$ satisfy the same inequality almost surely, they can be seen as the same variable. Thus, it is reasonable to set $\bar{\rho}$ equal to $\rho$. Finally, if we naively set $\zeta$ to zero, then observe that $\bar{\kappa} > \kappa$ for any values of $d \in \mathbb{N}$, $\delta \in (0,1)$, and $\rho \geq 1$, where the quantities $\bar{\kappa}$ and $\kappa$ are related by the equation $\bar{\kappa} = 2 \kappa + \rho \, 4 \log_{e}(2)$.

\begin{corollary}
If the relevant parameters of Lemma~\ref{lemma:AW_inequality} are set to those of Theorem~\ref{theorem:inequality_Z}, then there always exists a value of $\zeta$, where $\bar{\kappa} > \kappa$. Thus, in contrast to Lemma~\ref{lemma:SS_bounds}, the guarantees of Theorem~\ref{theorem:inequality_Z} require fewer number of sampled sensors.
\end{corollary}

The above analysis concludes that the smallest complexity bound $\kappa$ is achieved when $\zeta = 0$. Therefore, if one seeks to choose the smallest value of $\gamma$, then $\zeta$ should be set to zero and $\gamma$ set to $\lceil \kappa \rceil$. However, if $\gamma$ is some specified quantity and $\lceil \kappa \rceil$ (given $\zeta=0$) is smaller than $\gamma$, then the value of $\zeta$ can be adjusted to significantly tighten the semi-definite bounds in Theorem~\ref{theorem:inequality_Z}. We observe this in practice and we demonstrate it in Section~\ref{section:numerical_results}.

Furthermore, the above analysis on sampling complexity holds if one simply wants to satisfy the conditions required to use Theorem~\ref{theorem:inequality_Z}. 
However, another requirement may be that the lower bound in \eqref{eqn:inequality_theorem_Z_idd_c4} is non-trivial. If we also want to satisfy this requirement, then the sampling complexity now consists of choosing a value of $\gamma$ that simultaneously satisfies the following conditions, \eqref{eqn:complexity_00} and \eqref{eqn:complexity_02}. The complexity analysis for this scenario requires further investigation.

Finally, if the lower bound in \eqref{eqn:inequality_theorem_Z_idd_c4} is also required to be non-trivial, then we conjecture that Theorem~\ref{theorem:inequality_Z} still requires fewer number of sampled sensors than Lemma~\ref{lemma:AW_inequality} since the former is a generalization of the latter.


\subsection{Sensor Selection for State Estimation}
\label{subsection:application}

In this section, we use the matrix CIs of Section~\ref{subsection:matrix_CIs} to address the problem of sensor selection for state estimation via Kalman filtering. Specifically, we apply the CIs to the steady-state estimation error covariance matrix of the Kalman filter specified in Section~\ref{subsection:kalman_filtering}. 
In summary, our results in this section guarantee (in a probabilistic sense) what spectral properties the steady-state error covariance should possess when a selection $\mathcal{S}$ of sensors is randomly drawn according to the sampling policy in Section~\ref{subsection:sampling_policy}. 

Before we apply our CIs to the covariance equations of the Kalman filter, we lay out the conditions and assumptions that are necessary for our results in this section.


First, we set $\bar{\delta}$ of Lemma~\ref{lemma:AW_inequality} to $\delta$. This guarantees that our subsequent results, i.e., Lemma~\ref{lemma:SS_bounds} and Corollary~\ref{corollary:SS_bounds_novel}, can be properly compared. We define $\hat{\mathcal{T}} := (d,\delta,\gamma,\bar{\rho},\bar{\epsilon},p)$ as the tuple $\bar{\mathcal{T}}$ of Lemma~\ref{lemma:AW_inequality}, where $\bar{\delta}$ is set to $\delta$.

Second, our subsequent results assume that the pairs $(A, \mathbb{E}[ \boldsymbol{Z} ]^{1/2})$ and $(A,C_{\mathcal{S}})$ are detectable for any selection $\mathcal{S}$ of sensors. We refer the reader to Section~\ref{subsection:kalman_filtering} for the recursive form of the filtered error covariance $P_{\mathcal{S},(t)}$. Also, the steady-state solution to any quantity in this section is defined similar to how the steady-state solution of $P_{\mathcal{S},(t)}$ is defined in Section~\ref{subsection:kalman_filtering}.

Now, we introduce our first result. It is a reformulation of a theorem in \cite{10193790} that applies the guarantees of Lemma~\ref{lemma:AW_inequality} to obtain concentration-based guarantees for the steady-state error covariance. The theorem is Theorem~3 of \cite{10193790} and we refer it to as Lemma~\ref{lemma:SS_bounds} in this work. Our goal in this paper is to try to investigate whether Theorem~\ref{theorem:inequality_Z}, a generalization of Lemma~\ref{lemma:AW_inequality}, can lead to an improvement of Lemma~\ref{lemma:SS_bounds}.

\begin{lemma}
\label{lemma:SS_bounds}
Suppose the tuple $\hat{\mathcal{T}}$ of parameters satisfies the conditions required to use Lemma~\ref{lemma:AW_inequality}. Let $\bar{U}_{S}$ and $\bar{L}_{S}$ denote the unique p.d. steady-state solution to
\begin{gather}
\label{eqn:01}
\bar{U}_{S,(t+1)} = f( \, \bar{U}_{S,(t)} \, , \, ( \, 1 - \bar{\epsilon} \, ) \, \gamma \, \mathbb{E}[ \boldsymbol{Z} ] \, ),   \\
\bar{L}_{S,(t+1)} = f( \, \bar{L}_{S,(t)} \, , \, ( 1 + \bar{\epsilon} ) \, \gamma \, \mathbb{E}[ \boldsymbol{Z} ] \, ),   \nonumber
\end{gather}
respectively, such that $\bar{U}_{S,(-1)}, \bar{L}_{S,(-1)} \in \mathbb{S}_{+}^{d}$. Then,
\begin{align}
\label{eqn:CI_SS}
\mathbb{P} [ \, \bar{L}_S \preceq P_{\mathcal{S}} \preceq \bar{U}_S \, ] \geq ( 1 - \delta ).
\end{align}
\end{lemma}

Next, we obtain similar guarantees in Corollary~\ref{corollary:SS_bounds_novel} for the steady-state error covariance by applying Theorem~\ref{theorem:inequality_Z} instead of Lemma~\ref{lemma:AW_inequality}. We omit the bar notation from the bounds of Corollary~\ref{corollary:SS_bounds_novel} to distinguish them from those of Lemma~\ref{lemma:SS_bounds}.


\begin{corollary}
\label{corollary:SS_bounds_novel}
Suppose the tuple $\mathcal{T}$ of parameters satisfies the conditions required to use Theorem~\ref{theorem:inequality_Z}. Assume tuple $\mathcal{T}$ also satisfies the condition $( \, 1 - r \epsilon \, ) > 0$. Let $U_{S}$ and $L_{S}$ denote the unique p.d. steady-state solution to
\begin{gather}
\label{eqn:02}
U_{S,(t+1)} = f( \, U_{S,(t)} \, , \, ( \, 1 - r \epsilon \, ) \, \gamma \, \mathbb{E}[ \boldsymbol{Z} ] \, ),   \\
L_{S,(t+1)} = f( \, L_{S,(t)} \, , \, ( 1 + r \epsilon ) \, \gamma \, \mathbb{E}[ \boldsymbol{Z} ] \, ),   \nonumber
\end{gather}
respectively, such that $U_{S,(-1)}, L_{S,(-1)} \in \mathbb{S}_{+}^{d}$. Then,
\begin{align}
\label{eqn:CI_novel_01}
\mathbb{P} [ \, L_{S} \preceq P_{\mathcal{S}} \, ] &\geq ( \hspace{0.25mm} 1 - \delta \hspace{0.25mm} ), \\
\label{eqn:CI_novel_02}
\mathbb{P} [ \, P_{\mathcal{S}} \preceq U_{S} \, ] &\geq ( \hspace{0.25mm} 1 - \delta \hspace{0.25mm} ).
\end{align}
\end{corollary}

\begin{proof}

[Sketch] A sketch of the proof of Corollary~\ref{corollary:SS_bounds_novel} consists of applying Theorem~\ref{theorem:inequality_Z} instead of Lemma~\ref{lemma:AW_inequality} in a derivation similar to that of Lemma~\ref{lemma:SS_bounds}. One necessary condition for the upper bound $U_{S}$ to be defined is that $( \, 1 - r \epsilon \, ) \, \gamma \, \mathbb{E}[ \boldsymbol{Z} ]$ is non-zero and p.s.d.
\end{proof}

Now, we comment on the upper bounds of Lemma~\ref{lemma:SS_bounds} and Corollary~\ref{corollary:SS_bounds_novel}. We remind the reader that a quantity is non-trivial if it is non-zero and p.s.d. First, observe that the quantities $( \, 1 - \bar{\epsilon} \, ) \, \gamma \, \mathbb{E}[ \boldsymbol{Z} ]$ and $( \, 1 - r \epsilon \, ) \, \gamma \, \mathbb{E}[ \boldsymbol{Z} ]$ need to be necessarily non-trivial for $\bar{U}_S$ and $U_{S}$ to be defined, respectively. The bound $\bar{U}_S$ satisfies this necessary condition since any feasible tuple $\hat{\mathcal{T}}$ implies $(\,1-\bar{\epsilon}\,)>0$. However, the bound $U_S$ does not satisfy this necessary condition since a feasible tuple $\mathcal{T}$ does not imply $( \, 1 - r \epsilon \, ) > 0$. This is why we require the condition $( \, 1 - r \epsilon \, ) > 0$ in Corollary~\ref{corollary:SS_bounds_novel}.

The above discussion and our analysis on sampling complexity in Section~\ref{subsection:matrix_CIs} imply minimum requirements on the number of sampled sensors for $\bar{U}_S$ and $U_{S}$. Since Lemma~\ref{lemma:SS_bounds} is derived using Lemma~\ref{lemma:AW_inequality}, the sampling complexity of $\bar{U}_S$ consists of satisfying \eqref{eqn:complexity_01}. Furthermore, since the quantity $( \, 1 - r \epsilon \, ) \, \gamma \, \mathbb{E}[ \boldsymbol{Z} ]$ in \eqref{eqn:02} must be non-trivial, the sampling complexity for $U_S$ consists of satisfying \eqref{eqn:complexity_00} and \eqref{eqn:complexity_02}.


\section{Numerics}
\label{section:numerical_results}

In this section, we indirectly compare the guarantees of Lemma~\ref{lemma:AW_inequality} and Theorem~\ref{theorem:inequality_Z} by comparing the estimation performance guaranteed (in a probabilistic sense) by Lemma~\ref{lemma:SS_bounds} and Corollary~\ref{corollary:SS_bounds_novel}, respectively. We use the worst-case estimation performance $\bar{\lambda}(P_{\mathcal{S}})$ to gauge the quality of our state estimate, and we remind the reader that $\bar{\lambda}(P_{\mathcal{S}})$ is a random variable since it is randomly drawn via the sampling policy in Section~\ref{subsection:sampling_policy}. Since $\bar{\lambda}(P_{\mathcal{S}})$ is stochastic, we use Lemma~\ref{lemma:SS_bounds} and Corollary~\ref{corollary:SS_bounds_novel} to upper bound it. We denote their upper bounds as $\bar{\lambda}(\bar{U}_{S})$ and $\bar{\lambda}(U_{S})$, respectively. In this numerical study, we confirm that our generalized guarantees of Theorem~\ref{theorem:inequality_Z} outperform the existing ones of Lemma~\ref{lemma:AW_inequality} by showing that $\bar{\lambda}(U_{S})$ is a tighter bound on $\bar{\lambda}(P_{\mathcal{S}})$ than $\bar{\lambda}(\bar{U}_{S})$.

Before we compare $\bar{\lambda}(\bar{U}_{S})$ and $\bar{\lambda}(U_{S})$, we first state the assumptions we make. We assume $d = 3$, $\delta = 0.05$, $\eta = 420$, and $Q = 0.50 \, I_d$. Also, we assume the elements of state matrix $A$ are independently chosen at random from a uniform distribution in the interval $(0,1)$. Each element of the sequence $\bm{c}_i \in \mathbb{R}^{d}$ is similarly chosen for each candidate sensor, i.e., for all $i \in [ \eta ]$. Also, the measurement noise variance for each candidate sensor is identical, i.e., $\bm{\sigma}_i^2 = 0.50$ for all $i \in [ \eta ]$. We remind the reader that the support $\mathsf{S}_z$ of random variable $\boldsymbol{Z}$ can be computed given $\bm{c}_i$ and $\bm{\sigma}_i^2$ since $\mathcal{Z}_i := \bm{\sigma}_i^{-2} \bm{c}_i \bm{c}_i^T$ for all $i \in [\eta]$. To satisfy the detectability condition $(A,C_{\mathcal{S}})$ in Section~\ref{subsection:application}, we confirm that the pair $(A,\bm{c}_i)$ is detectable for all $i \in [ \eta ]$. Next, we delve into the procedure for computing $\bar{\lambda}(\bar{U}_{S})$ and $\bar{\lambda}(U_{S})$.

First, we execute Algorithm~5 of \cite{10193790} to find a tuple $\hat{\mathcal{T}}$ for the guarantees in Lemma~\ref{lemma:SS_bounds}. The algorithm outputs optimal values for the quantities $\rho$, $\epsilon$, and $p$. We denote them with an asterisk. We interpret the optimal value $p^{*}$ as the sampling distribution that minimizes the worst-case estimation performance of the steady-state error covariance. We denote this new tuple of parameters as $\hat{\mathcal{T}}^{*} := (d,\delta,\gamma,\rho^{*},\epsilon^{*},p^{*})$, and we compute the tuple $\hat{\mathcal{T}}^{*}$ for each $\gamma$ under consideration. Also, we satisfy the detectability condition $(A, \mathbb{E}[ \boldsymbol{Z} ]^{1/2})$ in Section~\ref{subsection:application} by confirming that $\mathbb{E}[ \boldsymbol{Z} ]$ is p.d. for every $p^{*}$.

Since there does not exist an algorithm for choosing a feasible tuple $\mathcal{T}$ of parameters for Theorem~\ref{theorem:inequality_Z}, the guarantees in Corollary~\ref{corollary:SS_bounds_novel} cannot be immediately compared to those of Lemma~\ref{lemma:SS_bounds}. However, since Theorem~\ref{theorem:inequality_Z} is a generalization of Lemma~\ref{lemma:AW_inequality}, the quantities outputted by Algorithm~5 of \cite{10193790} can be used to help choose a feasible $\mathcal{T}$ for Corollary~\ref{corollary:SS_bounds_novel}. Now, we outline a procedure for choosing a feasible tuple $\mathcal{T}$. First, we set $\rho$ to be the optimal value $\rho^{*}$ outputted by Algorithm~5 of \cite{10193790}. Next, we specify the scalar $\zeta \in [0,1]$. Then, we calculate $\epsilon$ using \eqref{eqn:inequality_theorem_Z_c3}. If $\epsilon$ resides in the interval $(0,2]$ as mandated by Theorem~\ref{theorem:inequality_Z}, then $\mathcal{T}^{*} :=(d,\delta,\gamma,\rho^{*},\epsilon,p^{*},\zeta)$ is feasible. If not, we try a different value of $\zeta$ until the $\epsilon$ we compute resides in the interval $(0,2]$. Though our procedure may result in a feasible $\mathcal{T}^{*}$, it does not guarantee that we also satisfy the condition $( \, 1 - r \epsilon \, ) > 0$. We compute the tuple $\mathcal{T}$ for each $\gamma$ and $\zeta$ under consideration. In Figure~\ref{figure:comparison}, we confirm that the condition $( \, 1 - r \epsilon \, ) > 0$ also holds for every value of $\gamma$ and $\zeta$ under consideration. As mentioned in Section~\ref{subsection:application}, we must satisfy the condition $( \, 1 - r \epsilon \, ) > 0$ for the upper bound $U_{S}$ in \eqref{eqn:CI_novel_02} to be necessarily defined.

Although the upper bound $\bar{\lambda}(\bar{U}_{S})$ cannot be computed for values of $\gamma$ that are below the minimum requirement $\bar{\kappa}$, this is not the case for $\bar{\lambda}(U_{S})$ since it is derived from Theorem~\ref{theorem:inequality_Z}. Unfortunately, since the parameters in $\hat{\mathcal{T}}^{*}$ are used to help choose a feasible tuple $\mathcal{T}^{*}$, the upper bound $\bar{\lambda}(U_{S})$ cannot be plotted for values of $\gamma$ that are below $\bar{\kappa}$. This is the reason why we do not consider values of $\gamma$ below $\bar{\kappa}$ in Figure~\ref{figure:comparison}b. We conjecture that an algorithm for Theorem~\ref{theorem:inequality_Z} and, as a consequence, Corollary~\ref{corollary:SS_bounds_novel}, exists for venturing significantly below $\bar{\kappa}$.

\begin{figure}
    \centering
    \includegraphics[width=\columnwidth]{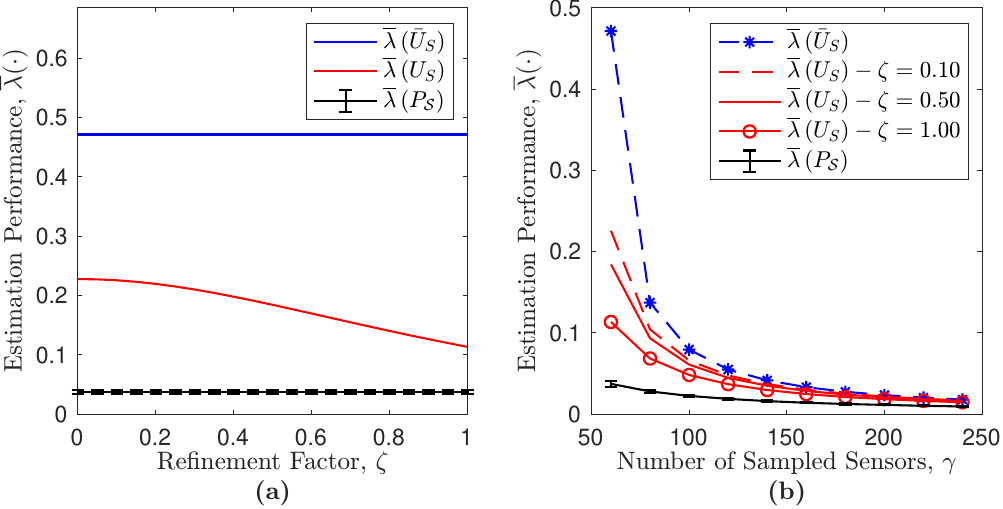}
    \caption{(a) Plot of the upper bounds on the worst-case estimation performance for varying values of the refinement factor $\zeta$. First, observe that $\overline{\lambda}( U_S )$ is a function of $\zeta \in [0,1]$. Also, observe that $\overline{\lambda}( \bar{U}_S )$ and $\overline{\lambda}( P_{\mathcal{S}} )$ are constant since they are not functions of $\zeta$. We remind the reader that $\overline{\lambda}( P_\mathcal{S} )$ is a random variable since it depends on a selection $\mathcal{S}$ that is randomly drawn via the sampling policy in Section~\ref{subsection:sampling_policy}. The average value of $\overline{\lambda}( P_{\mathcal{S}} )$ is indicated by the black curve and the variability of $\overline{\lambda}( P_\mathcal{S} )$ is captured by the standard deviation, where the error bars indicate $\pm$ one standard deviation. (b) Plot of the upper bounds on the worst-case estimation performance for varying number of sampled sensors. The quantity $\overline{\lambda}( U_S )$ is plotted for varying values of the refinement factor $\zeta$. Comments in subplot~(a), regarding the quantities $\overline{\lambda}( U_S )$ and $\overline{\lambda}( P_\mathcal{S} )$, also apply to subplot (b).
    }
    \label{figure:comparison}
\end{figure}

We refer the reader to Section~\ref{subsection:kalman_filtering} on how to the compute the steady-state solution of $P_{\mathcal{S},(t)}$. The filtered error covariance $P_{\mathcal{S},(t)}$ converges to its steady-state solution $P_{\mathcal{S}}$ in relatively few time steps.

Next, we show how to compute the semi-definite bounds $\bar{U}_{S}$ and $U_{S}$. We define $\bar{U}_{S} := \lim_{\, t \rightarrow \infty} \bar{U}_{S,(t)}$~and~$U_{S} := \lim_{\, t \rightarrow \infty} U_{S,(t)}$, where the sequences initially start at $\bar{U}_{S,(-1)}$ and $U_{S,(-1)}$, respectively. Similar to the filtered error covariance $P_{\mathcal{S},(t)}$, the quantities $\bar{U}_{S,(t)}$ and $U_{S,(t)}$ converge to their steady-state solution in relatively few time steps.

In Figure~1a, we compare the bounds $\bar{\lambda}(\bar{U}_{S})$ and $\bar{\lambda}(U_{S})$ for varying values of $\zeta$. We remind the reader that a smaller value equates to a tighter bound since we are comparing upper bounds for the worst-case estimation performance $\bar{\lambda}(P_{\mathcal{S}})$. This study is meant to observe the gradual effect that $\zeta$ has on improving the upper bound $\bar{\lambda}(U_{S})$. Also, we set $\gamma$ to be equal to the smallest value under consideration in Figure~\ref{figure:comparison}b, i.e., we  assume $\gamma = 60$. Figure 1a shows that the guarantees of Corollary~\ref{corollary:SS_bounds_novel} are shown to outperform those of Lemma~\ref{lemma:SS_bounds} with respect to the worst-case estimation performance. It also shows that $\bar{\lambda}(U_{S})$ achieves the tightest bound on $\bar{\lambda}(P_{\mathcal{S}})$ when $\zeta$ is set to its largest possible value, i.e., when $\zeta = 1$. We observe a similar trend in Figure~\ref{figure:comparison}b.

In Figure~1b, we compare the bounds $\bar{\lambda}(\bar{U}_{S})$ and $\bar{\lambda}(U_{S})$ for varying number of sampled sensors. The number ranges from $60$ to $240$. We also plot $\bar{\lambda}(U_{S})$ for varying values of $\zeta$ to observe its utility in tightening the bound. Figure~1b shows that the upper bound $\bar{\lambda}(U_{S})$ significantly outperforms $\bar{\lambda}(\bar{U}_{S})$ when the number of sampled sensors is relatively small. Similar to Figure~1a, Figure~1b also shows that $\bar{\lambda}(U_{S})$ is tighter for larger values of $\zeta$. In summary, Figure~1b shows that the guarantees of Corollary~\ref{corollary:SS_bounds_novel} outperform those of Lemma~\ref{lemma:SS_bounds} when we use the output of Algorithm~5 in \cite{10193790} to indirectly minimize our worst-case estimation performance.

We clarify that Algorithm~5 in \cite{10193790} was specifically designed to find a sampling distribution $p^{*}$ that indirectly minimized the bound $\bar{\lambda}(\bar{U}_S)$. We conjecture that the bound $\bar{\lambda}(U_{S})$ in our study can be further tightened by designing a procedure to find a sampling distribution that minimized $\bar{\lambda}(U_S)$. Such a procedure would also allow us to compute $\bar{\lambda}(U_S)$ for values of $\gamma$ that are significantly smaller than what is displayed in Figure~\ref{figure:comparison}b. Our analysis in Section~\ref{subsection:matrix_CIs} suggests it. Lastly, a large value of $\gamma$ suggests that either the candidate sensors under consideration are very noisy or a very good estimate of the state is required.





\section{Conclusion}
\label{section:conclusion}

In this work, we presented generalized CIs for a sum of i.i.d. and p.s.d. random matrices, and we applied them to the steady-state error covariance of the Kalman filter. Since our guarantees are a generalization of an existing matrix CI, we were able to show that our guarantees require significantly fewer number of sampled sensors. Also, we showed through a numerical analysis that our guarantees outperformed the existing ones in Lemma~\ref{lemma:AW_inequality}. When used to bound (in the semi-definite sense) the steady-state error covariance of the Kalman filter, we showed that our guarantees produced tighter bounds on the worst-case estimation performance.




A future direction of research consists of efficiently computing a tight upper bound $\bar{\lambda}(U_S)$ for smaller values of $\gamma$. This would require developing a polynomial-time algorithm similar to that of Algorithm~5 in \cite{10193790}.

\bibliographystyle{ieeetr}
\bibliography{references}

\appendix

We first establish several preliminary lemmas required for subsequent use in the proof of Theorem~\ref{theorem:inequality_Z}.

First, we introduce CIs for a sum of independent and symmetric random matrices. We formulate Lemma~\ref{lemma:AW_CI_01} for subsequent use in the proof of Theorem~\ref{theorem:inequality_Z}. We refer the reader to \cite{notes12} or Chapter 2 of \cite{qiu2014cognitive} for a proof of Lemma~\ref{lemma:AW_CI_01}.

\begin{lemma}
\label{lemma:AW_CI_01}
Let $( \boldsymbol{X}^{(i)} )_{\, i \in [\gamma]}$ denote a sequence of $\gamma$ independent and symmetric random matrices, i.e., $\boldsymbol{X}^{(i)} \in \mathbb{S}^{d}$ for all $i \in [\gamma]$. Define the random variable $\boldsymbol{ S_{\gamma} } := \textstyle\sum\nolimits_{\, i \in [\gamma]} \boldsymbol{X}^{(i)}$ and the events $\Theta_1 := \{ \, -\boldsymbol{ S_{\gamma} } \npreceq t \, I_d \, \}$ and $\Theta_2 := \{ \, \boldsymbol{ S_{\gamma} } \npreceq t \, I_d \, \}$.
Then, the following CIs,
\begin{gather}
\label{eqn:P_inequality_02}
\mathbb{P}[ \, \Theta_1 \, ] \leq d \, e^{-\lambda t} \, \textstyle\prod\nolimits_{\, i \in [\gamma]} \, \| \, \mathbb{E}[ \, e^{-\lambda \boldsymbol{X}^{(i)}} \, ] \, \| =: \delta_1, \\
\label{eqn:P_inequality_01}
\mathbb{P}[ \, \Theta_2 \, ] \leq d \, e^{-\lambda t} \, \textstyle\prod\nolimits_{\, i \in [\gamma]} \, \| \, \mathbb{E}[ \, e^{\, \lambda \boldsymbol{X}^{(i)}} \, ] \, \| =: \delta_2,
\end{gather}
hold for any scalars $\lambda > 0$ and $t > 0$.
\end{lemma}

Next, we establish Lemma~\ref{lemma:inequality_eX} and Lemma~\ref{lemma:inequality_Y_EY} for subsequent use in the proof of Lemma~\ref{lemma:inequality_X}. We omit the standard proof of Lemma~\ref{lemma:inequality_eX} for brevity.

\begin{lemma}
\label{lemma:inequality_eX}
If $X \in \mathbb{S}_{+}^{d}$ is a p.s.d. matrix, then $\| e^{X} \| = e^{ \| X \| }$.
\end{lemma}

\begin{lemma}
\label{lemma:inequality_Y_EY}
Let $\boldsymbol{Y} \hspace{-0.75mm} \in\hspace{-0.25mm} \mathbb{S}_{+}^{d}$ denote a p.s.d. random matrix. If the semi-definite inequality $\boldsymbol{Y} \hspace{-0.50mm} \preceq \hspace{-0.50mm} \rho \, I_d$ holds almost surely for a scalar $\rho > 0$, then $\| \, \boldsymbol{Y} - \mathbb{E}[ \boldsymbol{Y} ] \, \| \leq \rho$ almost surely.
\end{lemma}

\begin{proof}
Let $p_y \in \Delta^n$, $\mathsf{S}_y := \{ \mathcal{Y}_i \}_{\, i \in [n]}$, and $\mathbb{E}[ \boldsymbol{Y} ]$ denote the probability distribution, support, and expectation of random variable~$\boldsymbol{Y}$, respectively. Also, we define $\boldsymbol{\Phi} \hspace{-0.25mm} := \hspace{-0.25mm} ( \hspace{0.25mm} \boldsymbol{Y} - \mathbb{E}[ \boldsymbol{Y} ] \hspace{0.25mm} )$. First, observe that satisfying the inequality $\| \boldsymbol{\Phi} \| \leq \rho$ almost surely is equivalent to satisfying the following inequalities, $\boldsymbol{\Phi} \preceq \rho \, I_d$ and $\boldsymbol{\Phi} \succeq \rho \, I_d$, almost surely. The latter holds by assuming $\boldsymbol{Y} \preceq \rho \, I_d$ almost surely and employing the facts: (i)~$\textstyle\sum\nolimits_{\, i \in [n]} p_{y,i} = 1$, (ii)~$\mathbb{E}[ \boldsymbol{Y} ] = \textstyle\sum\nolimits_{\, i \in [n]} \, p_{y,i} \, \mathcal{Y}_i$, and (iii)~$\mathcal{Y}_i \succeq 0_d$ for all $i \in [n]$ since $\boldsymbol{Y} \in \mathbb{S}_{+}^{d}$.
\end{proof}

Next, we establish the scalar inequalities \eqref{eqn:inequality_00} and \eqref{eqn:inequality_01} in Lemma~\ref{lemma:inequality_X} for subsequent use in the proof of Lemma~\ref{lemma:inequality_Y}.

\begin{lemma}
\label{lemma:inequality_X}
Assume the random variable~$\boldsymbol{Y} \in \mathbb{S}_{+}^{d}$, a p.s.d. random matrix, satisfies the following inequalities, 
\begin{align}
\label{eqn:Y_rho}
\boldsymbol{Y} \preceq \rho \, I_d
\end{align}
almost surely for a scalar $\rho > 0$ and
\begin{align}
\label{eqn:inequality_L_Y_I}
\zeta \, \mathcal{I}_{\tilde{y}} \preceq \mathbb{E}[ \boldsymbol{Y} ] &\preceq \mathcal{I}_{\tilde{y}}
\end{align}
for a scalar $\zeta \in [0,1]$, where $\mathcal{I}_{\tilde{y}} := \tilde{Y} \tilde{Y}^{+}$ and $\tilde{Y} := \mathbb{E}[ \boldsymbol{Y} ]$. Assume the scalars $\rho$ and $\zeta$ satisfy the inequality $\rho \geq \zeta^2$. Define the random variable $\boldsymbol{X} := (\, \boldsymbol{Y} - \mathbb{E}[\boldsymbol{Y}] \,) \, / \, \rho$ and the scalar $\tilde{\rho} := ( \, 1 / \rho - \zeta^2 / \rho^2 \, )$. Then,
\begin{align}
\label{eqn:inequality_00}
\big{\|} \, \mathbb{E}[ \, e^{-\lambda \boldsymbol{X}} \, ] \, \big{\|} &\leq e^{\lambda^2 \tilde{\rho}}, \\
\label{eqn:inequality_01}
\big{\|} \, \mathbb{E}[ \, e^{ \, \lambda \boldsymbol{X} } \, ] \, \big{\|} &\leq e^{\lambda^2 \tilde{\rho}},
\end{align}
for any scalar $\lambda \in [0,1]$.
\end{lemma}

\begin{proof}
In this proof, we employ Lemma~\ref{lemma:inequality_eX} and Lemma~\ref{lemma:inequality_Y_EY} to derive inequality~\eqref{eqn:inequality_00}. The proof of inequality~\eqref{eqn:inequality_01} is omitted since it follows a similar derivation.

First, we observe that the quantity $\mathcal{I}_{\tilde{y}}$ can be factorized as $\mathcal{I}_{\tilde{y}} = Q_{\tilde{y}} \ \mathrm{diag}[ \ I_{\mathrm{rank}(\tilde{Y})} \, , \, 0_{( d - \mathrm{rank}(\tilde{Y}) )} \ ] \ Q_{\tilde{y}}^T$, where $Q_{\tilde{y}}$ is an orthogonal matrix. This is a consequence of the spectral decomposition theorem. An implication is that $\mathcal{I}_{\tilde{y}} \preceq I_d$. Also, observe that $\tilde{Y}$ can be factorized as $\tilde{Y} = Q_{\tilde{y}} \, D_{\tilde{y}} \, Q_{\tilde{y}}^T$, where $D_{\tilde{y}}$ is a diagonal matrix. Furthermore, if the matrix $\tilde{Y}$ is full rank, then $\mathcal{I}_{\tilde{y}}$ reduces to the identity matrix $I_d$.

Next, we specify three properties of $\boldsymbol{X}$. Note that Lemma~\ref{lemma:inequality_Y_EY} and the assumptions required to use Lemma~\ref{lemma:inequality_X} imply the~following, $\mathbb{E}[ \boldsymbol{X} ] = 0$, $\| \boldsymbol{X} \| \leq 1$, and $\mathbb{E}[ \boldsymbol{X}^2 ] \preceq \tilde{\rho} \, I_d$. The first implication follows from the definition of $\boldsymbol{X}$. The second implication follows from assuming $\rho > 0$, using the definition of $\boldsymbol{X}$, and employing Lemma~\ref{lemma:inequality_Y_EY}. The third implication follows from using the following, definitions, $\boldsymbol{X}$ and $\tilde{\rho}$, and employing the facts: (i) $\mathbb{E}[ \boldsymbol{Y}^2 ] \preceq \mathbb{E}[ \, \| \boldsymbol{Y} \| \, \boldsymbol{Y} \, ]$ and (ii) $\mathcal{I}_{\tilde{y}} \preceq I_d$. The third implication also assumes that \eqref{eqn:inequality_L_Y_I} and the inequality $\| \boldsymbol{Y} \| \leq \rho$ almost surely hold. Note that the assumption $\boldsymbol{Y} \preceq \rho \, I_d$ of Lemma~\ref{lemma:inequality_X} implies $\| \boldsymbol{Y} \| \leq \rho$ since $\boldsymbol{Y} \in \mathbb{S}_{+}^{d}$. Note also that the third property $\mathbb{E}[ \boldsymbol{X}^2 ] \preceq \tilde{\rho} \, I_d$ of random variable~$\boldsymbol{X}$ is valid if $\tilde{\rho} \geq 0$ and non-trivial if $\tilde{\rho} > 0$ $\Leftrightarrow$ $\rho > \zeta^2$. Thus, the condition $\rho > \zeta^2$ is a requirement.

Next, we specify properties about the matrix exponential function. If $X \in \mathbb{S}^d$ is a symmetric matrix, then
\begin{align}
\label{eqn:inequality_eX_01}
I_d + X \preceq e^{X}.
\end{align}
If the secondary condition $\| X \| \leq 1$ also holds, then
\begin{align}
\label{eqn:inequality_eX_02}
e^{X} \preceq I + X + X^2.
\end{align}
We refer the reader to \cite{notes13} for a proof of \eqref{eqn:inequality_eX_01} and \eqref{eqn:inequality_eX_02}.

Next, we employ the inequalities \eqref{eqn:inequality_eX_01} and \eqref{eqn:inequality_eX_02} to provide an upper bound on the expectation $\mathbb{E}[\, e^{-\lambda \boldsymbol{X}} \,]$, i.e.,
\begin{gather}
\mathbb{E}[\, e^{-\lambda \boldsymbol{X}} \,] 
\overset{(i)}{\preceq} \mathbb{E}[ \, I_d + (-\lambda) \boldsymbol{X} + (-\lambda)^2 \boldsymbol{X}^2 \, ] \nonumber \\
\label{eqn:inequality_eX_03}
\overset{(j)}{=} I_d + \lambda^2 \, \mathbb{E}[ \boldsymbol{X}^2 ] \overset{(k)}{\preceq} e^{\, \lambda^2 \, \mathbb{E}[\boldsymbol{X}^2]}
\end{gather}
where step~(i) and step~(k) employ \eqref{eqn:inequality_eX_02} and \eqref{eqn:inequality_eX_01}, respectively, step~(i) assumes $\| \lambda \boldsymbol{X} \| \leq 1$, and step~(j) employs the fact $\mathbb{E}[ \boldsymbol{X} ] = 0$. Note that the second property $\| \boldsymbol{X} \| \leq 1$ of random variable~$\boldsymbol{X}$ implies $\| \lambda \boldsymbol{X} \| \leq 1$ since $\lambda \in [0,1]$.

Finally, we obtain inequality~\eqref{eqn:inequality_00}, i.e.,
\begin{align*}
\big{\|} \, \mathbb{E}[ \, e^{-\lambda \boldsymbol{X}} \, ] \, \big{\|} 
\overset{(l)}{\leq} \big{\|} \, e^{\, \lambda^2 \, \mathbb{E}[\boldsymbol{X}^2]} \, \big{\|}
\overset{(m)}{=} e^{ \, \lambda^2 \, \| \mathbb{E}[\boldsymbol{X}^2] \| }
\overset{(n)}{\leq} e^{\, \lambda^2 \tilde{\rho}},
\end{align*}
where step~(l) employs inequality~\eqref{eqn:inequality_eX_03}, step~(m) employs the equality of Lemma~\ref{lemma:inequality_eX}, and step~(n) assumes $\| \hspace{0.25mm} \mathbb{E}[ \boldsymbol{X}^2 ] \hspace{0.25mm} \| \leq \tilde{\rho}$. Note that the third property $\mathbb{E}[ \boldsymbol{X}^2 ] \preceq \tilde{\rho} \, I_d$ of random variable~$\boldsymbol{X}$ implies $\| \hspace{0.25mm} \mathbb{E}[ \boldsymbol{X}^2 ] \hspace{0.25mm} \| \leq \tilde{\rho}$ since $\mathbb{E}[ \boldsymbol{X}^2 ] \in \mathbb{S}_{+}^{d}$
\end{proof}

Next, we derive CIs for a sum of independent and p.s.d. random matrices in Lemma~\ref{lemma:inequality_Y}.

\begin{lemma}
\label{lemma:inequality_Y}
Let $( \hspace{0.25mm} \boldsymbol{Y}_i \hspace{0.25mm} )_{\, i \in [\gamma]}$ denote a sequence of i.i.d. and p.s.d. random matrices, where $\boldsymbol{Y}_i$ is the $i$-th copy of random matrix $\boldsymbol{Y} \in \mathbb{S}_{+}^{d}$. Assume the sequence satisfies \eqref{eqn:Y_rho} almost surely and \eqref{eqn:inequality_L_Y_I} for the scalars $\rho > 0$ and $\zeta \in [0,1]$, and the scalars $\rho$ and $\zeta$ satisfy the inequality $\rho \geq \zeta^2$. If the equality~\eqref{eqn:inequality_theorem_Z_c3} holds for the scalars $\delta \in (0,1)$ and $\epsilon \in (0,2]$, then the events
\begin{gather}
\label{eqn:inequality_Y}
\left\{ \, \gamma \, \mathbb{E}[ \boldsymbol{Y} ] - \gamma \hspace{0.25mm} r \hspace{0.25mm} \epsilon \, I_d \preceq \textstyle\sum\nolimits_{\, i \in [\gamma]} \boldsymbol{Y}_i \, \right\}, \\
\label{eqn:inequality_Y_RHS}
\left\{ \, \textstyle\sum\nolimits_{\, i \in [\gamma]} \boldsymbol{Y}_i \, \preceq \, \gamma \hspace{0.25mm} r \hspace{0.25mm} \epsilon \, I_d + \gamma \, \mathbb{E}[ \boldsymbol{Y} ] \, \right\},
\end{gather}
each occur with probability at least $(1-\delta)$.
\end{lemma}

\begin{proof}
In this proof, we rely on Lemma~\ref{lemma:AW_CI_01} and Lemma~\ref{lemma:inequality_X} to show that the event \eqref{eqn:inequality_Y} occurs with probability at least $(1-\delta)$. The proof for event~\eqref{eqn:inequality_Y_RHS} is omitted since it follows a similar derivation. We assume the sequence $( \boldsymbol{X}^{(i)} )_{i \in [\gamma]}$ of Lemma~\ref{lemma:AW_CI_01} consists of i.i.d. and p.s.d. random matrices, where each matrix $\boldsymbol{X}^{(i)}$ in the sequence is a copy of the random variable $\boldsymbol{X}$ of Lemma~\ref{lemma:inequality_X}. We denote $\boldsymbol{X}_i$ and $\boldsymbol{Y}_i$ as the $i$-th copy of random variables $\boldsymbol{X}$ and $\boldsymbol{Y}$, respectively.

First, we employ Lemma~\ref{lemma:AW_CI_01} and Lemma~\ref{lemma:inequality_X} to obtain an upper bound on the probability $\mathbb{P} \left[ -\boldsymbol{ S_{\gamma} } \npreceq t \, I_d \, \right]$, i.e.,
\begin{gather}
\mathbb{P}\big{[} -\boldsymbol{S_{\gamma}} \npreceq t \, I_d \, \big{]}
\overset{(a)}{\leq} d \, e^{-\lambda t} \, \textstyle\prod\nolimits_{\, i \in [ \gamma ]} \| \, \mathbb{E}[ \, e^{-\lambda \boldsymbol{X}^{(i)}} \, ] \, \| \nonumber \\
\label{eqn:inequality_P_01}
\overset{(b)}{\leq} d \, e^{-\lambda t} \, \textstyle\prod\nolimits_{\, i \in [ \gamma ]} e^{\, \lambda^2 r / \rho }
\overset{(c)}{=} d \, e^{ -\lambda t + \lambda^2 \phi },
\end{gather}
where step~(a) restates \eqref{eqn:P_inequality_02} of Lemma~\ref{lemma:AW_CI_01}, step~(b) employs \eqref{eqn:inequality_00} of Lemma~\ref{lemma:inequality_X}, and step~(c) assumes $\phi := \gamma \hspace{0.25mm} r / \rho$. Note that the quantity $\lambda$ must satisfy the following conditions, $\lambda > 0$ and $\lambda \in [0,1]$, as a consequence of employing Lemma~\ref{lemma:AW_CI_01} and Lemma~\ref{lemma:inequality_X}, respectively. Thus, the quantity $\lambda$ must satisfy the following condition, $\lambda \in (0,1]$. Note that the guarantees \eqref{eqn:P_inequality_01} and \eqref{eqn:inequality_01} are employed instead of \eqref{eqn:P_inequality_02} and \eqref{eqn:inequality_00} for steps~(a) and (b), respectively, for the proof of event~\eqref{eqn:inequality_Y_RHS}.

Next, we minimize the RHS of step~(c) by substituting the value of $\lambda$ for the optimal value $\lambda^{*} = t \hspace{0.25mm} / \hspace{0.25mm} 2 \phi$, i.e., 
\begin{gather}
d \, e^{ -\lambda t + \lambda^2 \phi } \, |_{\, \lambda = \lambda^{*}} = d \, e^{ -\lambda^{*} \hspace{0.25mm} t + {\lambda^{*}}^2 \hspace{0.25mm} \phi } 
= d \, e^{ - t^2 / 4 \phi } \nonumber \\
\label{eqn:inequality_P_02}
\overset{(d)}{=} d \, e^{ - \epsilon^2 \phi \hspace{0.25mm} / \hspace{0.25mm} 4 }
\overset{(e)}{=} d \, e^{ - \epsilon^2 \gamma \, r \hspace{0.25mm} / \hspace{0.25mm} 4 \rho }
\overset{(f)}{=} \delta,
\end{gather}
where step~(d) assumes $t = \epsilon \hspace{0.25mm} \phi$ and $\epsilon \in (0,2]$, step~(e) employs the definition of $\phi$, and step~(f) assumes there exists a scalar $\delta \in (0,1)$ that is equal to the RHS of step~(e). Note that equation \eqref{eqn:inequality_theorem_Z_c3} of Theorem~\ref{theorem:inequality_Z} is an alternate formulation of the equality $\delta = d \, e^{ - \epsilon^2 \gamma \, r \hspace{0.25mm} / \hspace{0.25mm} 4 \rho }$. Also, the optimal value $\lambda^{*}$ is only defined if $\phi > 0$ $\Leftrightarrow$ $\rho > \zeta^2$ and satisfies the following, $\lambda^{*} = t \hspace{0.25mm} / \hspace{0.25mm} 2 \phi = \epsilon/2$. Thus, the condition $\epsilon \in (0,2]$ implies $\lambda^{*} \in (0,1]$ since $\phi > 0$. Also, observe that \eqref{eqn:inequality_P_01} and \eqref{eqn:inequality_P_02} imply the CI $\mathbb{P}\big{[} -\boldsymbol{ S_{\gamma} } \npreceq t \, I_d \, \big{]} \leq \delta$. Finally, the CI is equivalent to
\begin{align}
\label{eqn:inequality_03}
\mathbb{P}\big{[} -\boldsymbol{ S_{\gamma} } \preceq t \, I_d \, \big{]} \geq ( \hspace{0.25mm} 1 - \delta \hspace{0.25mm} ).
\end{align}

Finally, the event $\{ \, -\boldsymbol{ S_{\gamma} } \preceq t \, I_d \, \}$ under~consideration by the CI~\eqref{eqn:inequality_03} is equivalent to \eqref{eqn:inequality_Y} by using the definition of $\boldsymbol{X}$ and applying the fact $t = \epsilon \hspace{0.25mm} \phi = \gamma \hspace{0.25mm} r \hspace{0.25mm} \epsilon / \rho$.
\end{proof}

{\emph{Proof of Theorem 1:}} In this proof, we employ Lemma~\ref{lemma:inequality_Y} to derive CI \eqref{eqn:inequality_theorem_Z_idd_c4}. The proof of \eqref{eqn:inequality_theorem_Z_idd_c5} is omitted since the derivation is similar to that of \eqref{eqn:inequality_theorem_Z_idd_c4}. This proof references the quantities $\mathcal{I}_{\tilde{y}}$ and $\tilde{Y}$ of Lemma~\ref{lemma:inequality_X} since they are referenced in Lemma~\ref{lemma:inequality_Y}.

First, we define $\boldsymbol{Y}$ as the following, $\boldsymbol{Y} := \tilde{Z}^{+/2} \boldsymbol{Z} \, \tilde{Z}^{+/2}$. This definition bridges the connection between the random variables $\boldsymbol{Y}$ and $\boldsymbol{Z}$ of Lemma~\ref{lemma:inequality_Y} and Theorem~\ref{theorem:inequality_Z}, respectively.

Next, we define $\mathcal{I}_{\tilde{z}} := \tilde{Z} \tilde{Z}^{+}$ and $\tilde{Z} := \mathbb{E}[ \boldsymbol{Z} ]$ for random variable $\boldsymbol{Z}$. We define $\tilde{Z}^{+/2} := ( \tilde{Z}^{+} )^{1/2}$ for brevity. Also, observe that the quantity $\mathcal{I}_{\tilde{z}}$ can be factorized as $\mathcal{I}_{\tilde{z}} = Q_{\tilde{z}} \ \mathrm{diag}[ \ I_{\mathrm{rank}(\tilde{Z})} \, , \, 0_{( d - \mathrm{rank}(\tilde{Z}) )} \ ] \ Q_{\tilde{z}}^T$, where $Q_{\tilde{z}}$ is an orthogonal matrix. This is a consequence of the spectral decomposition theorem. Also, observe that $\tilde{Z}$ can be factorized as $\tilde{Z} = Q_{\tilde{z}} \, D_{\tilde{z}} \, Q_{\tilde{z}}^T$, where $D_{\tilde{z}}$ is a diagonal matrix. The~following statements,
\begin{gather}
\label{eqn:fact_01}
\mathcal{I}_{\tilde{z}} = \tilde{Z}^{+/2} \tilde{Z} \hspace{0.25mm} \tilde{Z}^{+/2} = \tilde{Z} \tilde{Z}^{+} \tilde{Z} \tilde{Z}^{+}, \\
\label{eqn:fact_02}
\mathcal{I}_{\tilde{z}} = \tilde{Z}^{1/2} \tilde{Z}^{+/2} = \tilde{Z}^{+/2} \tilde{Z}^{1/2}, \\
\label{eqn:fact_03}
\mathcal{I}_{\tilde{y}} = \tilde{Z}^{+/2} \tilde{Z} \hspace{0.25mm} \tilde{Z}^{+/2},
\end{gather}
are due to the application of the pseudo-inverse to symmetric matrices. Since $\tilde{Z}$ is symmetric, the following also holds, $\tilde{Z}^{+} = ( Q_{\tilde{z}} \, D_{\tilde{z}} \, Q_{\tilde{z}}^T )^{+} = Q_{\tilde{z}} D_{\tilde{z}}^{+} Q_{\tilde{z}}^T$. Note that \eqref{eqn:fact_03} follows from repeatedly employing the fact \eqref{eqn:fact_01} and using the definitions of $\tilde{Y}$, $\boldsymbol{Y}$ and $\tilde{Z}$.

Next, we show that \eqref{eqn:inequality_R} of Theorem~\ref{theorem:inequality_Z} implies \eqref{eqn:Y_rho}. We prove this by first pre- and post- multiplying \eqref{eqn:inequality_R} by $\tilde{Z}^{+/2}$. Then, we use the definitions of $\boldsymbol{Y}$ and $\tilde{Z}$ to simplify the former. We conclude by employing \eqref{eqn:fact_03} and $\mathcal{I}_{\tilde{y}} \preceq I_d$.

Next, we show that \eqref{eqn:inequality_L_Y_I} is implied by the assumptions of Theorem~\ref{theorem:inequality_Z} as a consequence of the definition of $\boldsymbol{Y}$. We prove this by first recognizing that the following semi-definite inequality, $\zeta \, \tilde{Z} \preceq \tilde{Z} \preceq \tilde{Z}$, is true since $\zeta \in [0,1]$. Next, we pre- and post- multiply the inequality by $\tilde{Z}^{+/2}$. Then, we use the definitions of $\boldsymbol{Y}$ and $\tilde{Z}$ to simplify the former. We conclude by employing the fact \eqref{eqn:fact_03}.

Next, observe that the events \eqref{eqn:inequality_Y} and \eqref{eqn:inequality_Y_RHS} each occur with probability $( 1 - \delta )$ since the assumptions of Theorem~\ref{theorem:inequality_Z} imply the inequalities \eqref{eqn:Y_rho} and \eqref{eqn:inequality_L_Y_I} of Lemma~\ref{lemma:inequality_Y}.

Next, we show that \eqref{eqn:inequality_Y} implies the event under consideration by the CI \eqref{eqn:inequality_theorem_Z_idd_c4} of Theorem~\ref{theorem:inequality_Z}. We prove this by first using the definition of $\boldsymbol{Y}$ and expressing \eqref{eqn:inequality_Y} in terms of $\tilde{Z}$ and $\mathbf{Z}$. Next, we pre- and post- multiply the former by $\tilde{Z}^{1/2}$ and employ the fact~\eqref{eqn:fact_02}. We conclude by employing the facts $\mathbb{E}[ \boldsymbol{Z} ] = \mathcal{I}_{\tilde{z}} \, \mathbb{E}[ \boldsymbol{Z} ] \, \mathcal{I}_{\tilde{z}}$ and $\boldsymbol{Z}_i = \mathcal{I}_{\tilde{z}} \, \boldsymbol{Z}_i \, \mathcal{I}_{\tilde{z}}$. We refer the reader to Section~\ref{subsection:properties} for the definition of $\boldsymbol{Z}_i$. Note that $\boldsymbol{Z}_i = \mathcal{I}_{\tilde{z}} \, \boldsymbol{Z}_i \, \mathcal{I}_{\tilde{z}}$ since $\mathrm{im}( \boldsymbol{Z}_i ) \subseteq \mathrm{im}( \tilde{Z} )$. This completes the proof of Theorem~\ref{theorem:inequality_Z}.

Finally, we show that the assumption $\rho \geq 1$ is necessary. First, observe that satisfying \eqref{eqn:inequality_R} almost surely is equivalent to satisfying the inequality $\rho \, \mathbb{E}[ \boldsymbol{Z} ] - \mathcal{Z}_i \succeq 0_d$ for all $i \in [ \eta ]$. Next, we multiply $p_i$ to both sides of the $i$-th inequality for all $i \in [ \eta ]$. We conclude by summing the $\eta$ inequalities into one and observing that the assumption $\rho \geq 1$ is necessary.



\end{document}